\documentclass[letterpaper, 10 pt, conference]{ieeeconf}  % Comment this line out if you need a4paper

\IEEEoverridecommandlockouts                              % This command is only needed if 
\usepackage{graphics} % for pdf, bitmapped graphics files
\usepackage{epsfig} % for postscript graphics files
\usepackage{mathptmx} % assumes new font selection scheme installed
\usepackage{times} % assumes new font selection scheme installed
\usepackage{amsmath} % assumes amsmath package installed
\usepackage{amssymb}  % assumes amsmath package installed
\usepackage{enumerate}% http://ctan.org/pkg/enumerate
\usepackage{algorithm}
\usepackage[noend]{algpseudocode}
\usepackage{soul}
\usepackage{textcomp}
\usepackage{graphicx}      % include this line if your document contains figures
\usepackage{cite}
\usepackage{xcolor}
\usepackage{stackengine}
\usepackage{tikz} 
\usetikzlibrary{shapes.geometric}
\usetikzlibrary{arrows.meta,arrows}
\newtheorem{remark}{Remark} %[in_counter]

\title{\LARGE \bf
Information-Constrained Optimal Control of Distributed Systems with Power Constraints
}

\author{V.~Causevic$^{\dagger}$, P.~Ugo~Abara$^{\dagger}$ and S.~Hirche$ $% <-this % stops a space
\thanks{This project has received funding from European Union's Horizon 2020
Framework Programme for Research and Innovation under grant agreement
No 674875 and the German Research Foundation (DFG) within the Priority Program
SPP 1914 ''Cyber-Physical Networking".}% <-this % stops a space
\thanks{$^{\dagger}$ Both authors contributed equally to this work.}%
\thanks{$^{1}$ V.~Causevic, P.~Ugo Abara and S. Hirche are with the Chair of Information-oriented
Control, Technical University of Munich, Germany; {\tt\small http://www.itr.ei.tum.de, \{vedad.causevic, ugoabara, hirche\}@tum.de}}%
}

\newcommand\numberthis{\addtocounter{equation}{1}\tag{\theequation}}
\newtheorem{theorem}{Theorem} %[in_counter]
\newtheorem{lemma}{Lemma} %[in_counter]
\newtheorem{proposition}{Proposition} %[in_counter]
\newtheorem{corollary}{Corollary} %[in_counter]
\newtheorem{definition}{Definition} %[in_counter]
\makeatletter
\def\BState{\State\hskip-\ALG@thistlm}
\makeatother
\begin{document}

\maketitle
\thispagestyle{empty}
\pagestyle{empty}

%%%%%%%%%%%%%%%%%%%%%%%%%%%%%%%%%%%%%%%%%%%%%%%%%%%%%%%%%%%%%%%%%%%%%%%%%%%%%%%%
\begin{abstract}

In this paper we address the problem of information-constrained optimal control for an interconnected system subject to one-step communication delays and power constraints. The goal is to minimize a finite-horizon quadratic cost by optimally choosing the control inputs for the subsystems, accounting for power constraints in the overall system and different information available at the decision makers. To this purpose, due to the quadratic nature of the power constraints, the LQG problem is reformulated as a linear problem in the covariance of state-input aggregated vector. The zero-duality gap allows us to equivalently consider the dual problem, and decompose it into several sub-problems according to the information structure present in the system. Finally, the optimal control inputs are found in a form that allows for offline computation of the control gains.   

\end{abstract}

%%%%%%%%%%%%%%%%%%%%%%%%%%%%%%%%%%%%%%%%%%%%%%%%%%%%%%%%%%%%%%%%%%%%%%%%%%%%%%%%
\section{INTRODUCTION}

Technological advances in computation and communication, and societal needs have revived
the research interest in control of interconnected systems
 \cite{networkflows}. Some examples include smart grids, communication networks, and transportation systems. Traditionally, arguments in favor of distributed control (compared to centralized) are geographically distributed sensors, limited local computational power at the plant side, robustness against single-node failure and information privacy.%For instance, in a building district, computing devices and sensors are physically dislocated. Additionally, individual buildings might not be willing to share their electricity consumption profiles \cite{buildings}. 
\\
\indent
In general, the design of distributed control is difficult because it imposes information constraints on individual decision makers. Such constraints arise due to either partial information exchange between decision makers or communication delay. In the problem we address herein, decision makers are able to communicate the full information they receive - either due to own measurements or from other decision makers, however, with delay. In other words, information constraints are due to communication delays between decision makers. The information constraints, sometimes referred to as information structure, play a key role in determining the optimal control and decide on its computational tractability. Indeed, in \cite{witsenhausen1968} a linear quadratic Gaussian team problem is constructed with a non-classical information pattern and it is shown that a linear controller is not necessarily optimal. This problem is addressed in \cite{ho-chi1972} where it is shown that the so-called partially nested information structure guarantees existence of optimal control laws that are linear in the associated information. Finally, a strong result characterizing the class of all information-constrained problems which may be cast as a convex program is given in \cite{rotkowitz2006tac}. \\
Inspiration for our approach is given by the work in \cite{nayyar2013tac} which suggests that the information hierarchy existing between the decision makers can be exploited to obtain the optimal solution. First explicit solutions to linear quadratic Gaussian team problems that adopt similar approach are given in \cite{lamperski2012cdc}. The authors however, consider a typical unconstrained linear quadratic team problem. But in reality, e.g. actuation capabilities are limited and thus must be accounted for in the design procedure. \\
\indent The main contribution of this paper is a method to compute optimal control laws, for a power-constrained system with  given information structure. We assume the latter to be induced by a one-step communication delays between the decision makers. To this end, the problem is reformulated in its dual Lagrangian form, where the covariance of the state-input aggregated vector is defined as decision variable. The information structure is then exploited to split the optimization problem into simpler sub-problems that have alike structure. Indeed, in-network control \cite{InNetwork} is seen as the decomposition of a complex task into smaller sub-tasks resulting in computationally inexpensive local control actions. From an application point of view, the goal is to implement and analyze the developed approach within a network infrastructure, exploiting the possibility of existing (but limited) in-network processing, in order to improve control performance.\\
\indent The remainder of the paper is outlined as follows. We start with problem setup in \ref{sec: problem statement}. The method to decouple problem into several sub-problems via covariance decomposition is presented in section \ref{sec:info dec}. In section \ref{sec:dual} we provide structural characterization of the solution to the problem and finally conclusions are given in \ref{sec:conclude}.
\section{Problem setting}
\label{sec: problem statement}
Consider a large-scale interconnected dynamical system composed of $N$ physically-coupled linear time-invariant (LTI) subsystems. Formally, the physical interconnections are described through a graph $\mathcal{G} = \left(\mathcal{V}, \mathcal{E}\right)$. We will refer to it as the physical interconnection graph. Each node $i \in \mathcal{V}$ corresponds to one of the subsystems $i\in\{1, \ldots, N\}$. An edge $(j,i) \in \mathcal{E}$ if dynamics of node $i$ is directly affected by node $j$. We assume that $\mathcal{G} $ is connected and undirected, i.e., $(i,j) \in \mathcal{E}$ if and only if $(j,i) \in \mathcal{E}$. The set of direct neighbors of decision maker $i$ is defined as $\mathcal{N}_i = \{j \,\vert (j,i) \in \mathcal{E}\}$. The length of the shortest path between nodes $i$ and $j$ will be denoted by $d_{ij}$. Clearly, if $j \in \mathcal{N}_i$ then $d_{ij} = 1$. The dynamics of the $i$-th subsystem is given by a first order stochastic difference equation
\begin{align*}
	&x_i (k+1)= A_i x_i (k)+ B_i u_i (k) + \sum_{j \in \mathcal{N}_{i}} A_{ij} x_j (k)+ w_i (k),
	\label{eq: main NCS}
	\numberthis
\end{align*}
where $A_i \in \mathbb{R}^{n_i \times n_i}$,  $A_{ij} \in \mathbb{R}^{n_i \times n_j}$,  $B_i \in \mathbb{R}^{n_i \times m_i}$, $ x_i (k) \in \mathbb{R}^{n_i}$ is the state and $u_i (k) \in \mathbb{R}^{m_i}$ is the control signal of the $i$-th subsystem. The noise process $w_i (k) \in \mathbb{R}^{n_i}$ is zero-mean i.i.d. Gaussian noise with covariance matrix $\Sigma_{w} $. The initial state $x_i (0)$ is a random variable with zero-mean and finite covariance $\Sigma_{x}  $. Moreover, $x_i (0)$  and $w_i (k)$ are assumed to be pair-wise independent at each time instant $k$ and every $i$. For a more compact notation, equation \eqref{eq: main NCS} can be rewritten as
\begin{equation}
x (k+1) = A x(k) + B u(k) + w(k)
    \label{eq: global system}
\end{equation}
where the stacked vectors are $x (k) = ({x_1 ^\top (k)}, \ldots , {x_N ^\top (k)})^\top \in \mathbb{R}^n$, $ w(k) = ({w_1 ^\top (k)}, \ldots,  {w_N ^\top (k)})^\top \in \mathbb{R}^n$, $ u(k) = ({u_1 ^\top (k)}, \ldots,  {u_N ^\top (k)})^\top \in \mathbb{R}^m$, $n=\sum_{i=1}^{N} n_i$ and $m=\sum_{i=1}^{N} m_i$. The admissible control policies at time instant $k$ are measurable functions of the information available to each decision maker $i$ (sometimes also referred to as player $i$)
\begin{equation}
    u_i (k) = \gamma_k^i(\mathcal{I}_k^i)
    \label{eq:gama}
\end{equation}
where $\mathcal{I}_k^i, \ k=0,\ldots,T-1,$ is defined as
\begin{align*}
    &\mathcal{I}_k^i = \{\mathcal{I}_{k-1}^i, x^i_{k}, u_{k-1}^i\} \underset{j \in \mathcal{N}_i}\bigcup \{\mathcal{I}^j_{{k-1}}\}, \quad k>0,  
    \numberthis\label{eq:information set}
\end{align*}
and $\mathcal{I}_0^i=\lbrace x_0 ^i \rbrace$. In other words, the information set of each decision maker $i$ is updated at time instant $k$ by the current state and the one-step delayed information from the direct neighbors $\mathcal{N}_i$. The objective is to minimize the following global control cost
\begin{align*}
    \numberthis \label{eq: quadratic cost 1}
    J_{\mathcal{C}} = {\rm E}\left[
        \sum_{k=0}^{T-1} {
            \begin{bmatrix} 
                x(k) \\ 
                u(k) 
            \end{bmatrix}}^\top Q   
            \begin{bmatrix} 
                x(k) \\ 
                u(k) 
            \end{bmatrix} + x(T)^\top Q_{T} x(T) 
        \right]
\end{align*}
where the matrix $Q$ is partitioned according to the vector $z (k) $ = $\left[{x(k)}^\top {u(k)}^\top\right]^\top$  i.e.
\begin{equation*}
\numberthis
\label{eq:matrixQ}
    Q= \begin{bmatrix}Q_{xx} & Q_{xu} \\ Q_{ux} & Q_{uu} \end{bmatrix}.
\end{equation*}
The matrix $Q_{uu}$ is assumed to be positive-definite matrix, while $Q$ and $Q_T$ are assumed to be semi-definite positive. We also assume controllability of pair (A,B) as well as detectability of $(Q^{\frac{1}{2}},A)$. Moreover, it is assumed that each decision maker knows the parameters of the overall system. \\
The cost \eqref{eq: quadratic cost 1} is to be minimized under power constraints, which are defined as
\begin{align*}
    \numberthis
    & {\rm E}\left[ {z(k)}^\top W_i \, z (k) \right] \leq p_k^i, \quad \forall i= 1, \ldots, M
    \label{eq: power constraints}
\end{align*}
where $W_i \in \mathbb{R}^{(n+m) \times (n+m)}$, $i=1,\ldots,M$, is a positive semi-definite weighting matrix. By appropriate choice of $W_i$, the set of constraints in \eqref{eq: power constraints}  captures either constraints present in the power of the overall system, or those related to the individual subystems. Ultimately, the problem is formally stated as
% \begin{align*}
%     \numberthis \label{eq: problem1}
%     \min_{\gamma_{0:T-1}} \quad & {\rm E}\left[
%             \sum_{k=0}^{T-1} {
%                 \begin{bmatrix} 
%                     x(k) \\ 
%                     u(k) 
%                 \end{bmatrix}}^\top Q  
%                 \begin{bmatrix} 
%                     x(k) \\ 
%                     u(k) 
%                 \end{bmatrix} + x(T)^\top Q_{T} \, x(T) 
%             \right]\\
%     \text{s.t.} \quad &  x (k+1) = A x(k) + B u(k) + w(k) \\
%     & {\rm E}\left[
%         {\begin{bmatrix} 
%             x(k) \\ 
%             u(k) 
%         \end{bmatrix}}^\top W_i \begin{bmatrix} 
%                                     x(k) \\ 
%                                     u(k) 
%                                 \end{bmatrix} 
%         \right] \leq p_k^i, \quad \forall i= 1, \ldots, M\\
%     & u_j (k) = \gamma_k^j(\mathcal{I}_k^j), \quad \forall j= 1, \ldots, N
% \end{align*}
\begin{align*}
    \numberthis \label{eq: problem1}
    \min_{\gamma_{0:T-1}} \qquad & {\rm E}\left[
            \sum_{k=0}^{T-1} {
                \begin{bmatrix} 
                    x(k) \\ 
                    u(k) 
                \end{bmatrix}}^\top Q  
                \begin{bmatrix} 
                    x(k) \\ 
                    u(k) 
                \end{bmatrix} + x(T)^\top Q_{T} \, x(T) 
            \right]\\
    \text{s.t.} \qquad & \eqref{eq: global system},  \eqref{eq:gama}, \eqref{eq: power constraints}
\end{align*}
where $\gamma_k = [ \gamma^{1}_{k},\ldots,\gamma^{N}_{k}]$ is composed of all players control policies. Before stating the main result of this section we define the notion of partial nestedness \cite{PN}.
\begin{definition}
The information structure $\mathcal{I}_{k} = \left\lbrace  \mathcal{I}_{k} ^1, \ldots, \mathcal{I}_{k} ^N \right\rbrace$ is partially nested if, for every admissible policy \eqref{eq:gama}, whenever $u_i(\tau)$ affects $\mathcal{I}_{k} ^j$, then $\mathcal{I}_{\tau} ^i \subset \mathcal{I}_{k} ^j$.
\end{definition}
\begin{lemma}[Partial nestedness]
The information structure defined by \eqref{eq:information set} is partially nested.
\label{lemma: partially nested}
\end{lemma}
\begin{proof}
    Let $d_{ji}$ be the length of shortest path  $j \rightarrow i$ in the physical interconnection graph. Considering \eqref{eq:information set}, the information set $\mathcal{I}^i_k$ is influenced by measurement $x_j (k-d_{ji})$, or equivalently by $u_j (k-d_{ji}-1)$. Thus, to check if information structure \eqref{eq:information set} is partially nested, one should verify the condition: $\mathcal{I}_{k-d_{ji}-1}^j \subset \mathcal{I}_k^i$. Recalling the assumption that graph $\mathcal{G}$ is connected and undirected, the information sets of decision makers $i$ and $j$ are explicitly written as
    \begin{align*}
    \mathcal{I}_k^i= &\underset{n = 1,\ldots, N}\bigcup \left\lbrace x_n (0:k-d_{ni}) \right\rbrace, \\
    \mathcal{I}_{k-d_{ji}-1}^j =& \underset{n = 1,\ldots, N}\bigcup \left\lbrace x_n (0:k-d_{nj}-d_{ji}-1)\right\rbrace,
      \end{align*}    
    which reduces the partial nestedness condition to: $d_{nj}+d_{ji}+1 \geq d_{ni} $. Since $d_{ni}$ is the length of the shortest path between nodes $n$ and $i$  in $\mathcal{G}$, one can write: $d_{ni} \leq d_{nj} + d_{ji} < d_{nj}+d_{ji}+1 $ which concludes the proof.
\end{proof}
\indent \\
\indent Taking into consideration that problem \eqref{eq: problem1} is subject to power constraints, it is convenient to reformulate it in terms of covariance as the new decision variable
\begin{equation*}
    V(k)= {\rm E} \left[ z(k) z(k)^\top \right] ={\rm E} \left[ {\begin{bmatrix} 
                x(k) \\ 
                u(k) 
            \end{bmatrix}} \begin{bmatrix} 
                                        x(k) \\ 
                                        u(k) 
                                    \end{bmatrix}^\top\right]
\end{equation*}
With the additional constraint given by \eqref{eq:gama}, problem \eqref{eq: problem1} is posed as a covariance selection problem
\begin{align*}
    \numberthis\label{eq: problem2}
        \min_{V(0:T-1)\succeq 0} \quad &  tr (Q_{T} V_{xx} (T)) +\sum_{k=0}^{T-1} tr (Q V(k))   \\
        \text{s.t.} \quad & F V(0) F^{\top} = \Sigma_x \\ 
        \quad &  \begin{bmatrix} A & B  \end{bmatrix} V(k) \begin{bmatrix} A & B  \end{bmatrix}^\top + \Sigma_w  = F V(k+1) F^\top \\
        & tr(W_i V(k)) \le p_k ^{i} , \quad \forall i= 1, \ldots, M 
\end{align*}
where $F= \begin{bmatrix} I & 0 \end{bmatrix}$. Part of the result above is derived from the fact that, for a generic matrix $\Theta$ the following identity holds
\begin{align*}
    {\rm E}\left[
        z(k)^\top \Theta z(k)
        \right] = tr\left( \Theta V(k) \right).
\end{align*}
Additionally, rewriting the system dynamics equation \eqref{eq: global system} in terms of a covariance variable $V$  
\begin{align*}
	& F {V}(k+1)F^\top = V_{xx} (k+1)={\rm E}\left[ x(k+1) x(k+1)^\top \right] \\
	& = \begin{bmatrix} A & B  \end{bmatrix} {V}(k)\begin{bmatrix} A & B  \end{bmatrix}^\top + \Sigma_w  
\end{align*}
and translating the initial condition ${\rm E}\left[ x(0) x(0)^\top \right]  =  \Sigma_x $ into  
    \begin{align*}
	V_{xx} (0)&= {\rm E}\left[ x(0) x(0)^\top \right]  = F {V}(0)F^\top = \Sigma_x ,
\end{align*}
the form in \eqref{eq: problem2} is obtained.
\section{Information decomposition}
\label{sec:info dec}
\subsection{Covariance Decomposition}
\label{par: covariance decompostion}
For the sake of simplicity of derivation we demonstrate the method on a two-player system. Considering the state equation \eqref{eq: global system}, each decision maker at each time instant $k$ is able to compute the estimate of the state $x$ based on the common information $\mathcal{I}_k^0$ the two players have at time instant $k$, i.e.
\begin{align*}
    \numberthis \label{eq:common information set}
    \mathcal{I}_k^0  = \mathcal{I}_k^1 \cap \mathcal{I}_k^2 = \{x({0:k-1}), u({0:k-1})\},
\end{align*}
later referred to as the coordinator's information set. The estimate is given by
\begin{align*}
    \numberthis \label{eq:estimator}
    & \hat{x} (k)  ={\rm E}\left[x(k) |\mathcal{I}_k^0 \right] = A x (k-1) + B u (k-1),  
\end{align*}
since ${\rm E}\left[w(k-1) |\mathcal{I}_k^0 \right] = 0$. Locally, after measuring its own state $x_i$ each decision maker can compute the local noise value at the previous time step as
\begin{align*}
\numberthis
& \omega_{i} (k)  = w_{i} (k-1) = x_i (k) - M_i^\top\hat{x} (k)
\label{eq:noise1}
\end{align*}
where $M_1^\top = \begin{bmatrix} I & 0  \end{bmatrix}^\top$, $M_2^\top = \begin{bmatrix} 0 & I  \end{bmatrix}^\top$.
    
The quantities $\hat{x} , \omega_{1}, \omega_{2}$ form a pair-wise independent components of state.  Due to linearity of the state decomposition given by \eqref{eq:estimator}, \eqref{eq:noise1} and partial nestedness of the information structure \eqref{eq:information set} one can represent the optimal control input in the form
\begin{align*}
    \numberthis\label{eq:input decomposition}
    & u(k) = \hat{\phi} (k)+\begin{bmatrix} \phi_1 (k) \\  \phi_2 (k)  \end{bmatrix}
    \end{align*}
where $\hat{\phi} (k) = -{L}_0 (k) \hat{x} (k)$, $\phi_1 (k) = - L_1 (k) \omega_1 (k)$ and $\phi_2 (k) = - L_2 (k) \omega_2 (k)$, for some gains $L_0, L_1, L_2$.
Aiming for the decomposition of problem \eqref{eq: problem2}, we define a vector $\bar {z}$ of state components $\hat{x},\omega_{1},\omega_{2}$ and input components $\hat{\phi},\phi_{1},\phi_{2}$  
\begin{equation*}
\numberthis \label{eq:zbar}
    \bar {z}(k) = \begin{bmatrix} 
                \hat {x}(k)^\top & \hat {\phi}(k)^\top \vert \omega_1(k)^\top & \phi_1(k)^\top \vert \omega_2(k) ^\top & \phi_2(k) ^\top
                \end{bmatrix} ^\top
\end{equation*}
whose blocks are independent.
{Additionally, denoting the state decomposition  \eqref{eq:estimator} - \eqref{eq:noise1} and  input decomposition in \eqref{eq:input decomposition} as
\begin{align*}
(x^{0} (k) ,x^{1} (k) , x^{2} (k))=(\hat{x}(k),\omega_1 (k), \omega_2 (k)),\\
 (u^{0} (k) ,u^{1} (k) , u^{2} (k))=(\hat{\phi}(k),\phi_1 (k), \phi_2 (k)).
\end{align*} the covariance matrix of $\bar{z}(k)$ 
is given by
\begin{equation}
    \label{eq: augmented state variance}
    \bar{V}(k) = {\rm E} \left[\bar{z}(k)\bar{z}(k)^\top\right] = \left[
    \begin{array}{ccc}
    {V}^0(k) & 0 & 0 \\
    0 & V^1(k) & 0 \\ 
    0 & 0 & V^2(k)
    \end{array}\right]
\end{equation}
where covariance matrices $V^l, l \in \{0,1,2\},$ of the individual blocks are
\begin{align*}
    %\numberthis\label{eq: covariances}
    & {V}^l (k) = {\rm E}\left[ \begin{bmatrix} {x}^l(k) \\ u^l (k) \end{bmatrix}\begin{bmatrix} x(k) \\ u(k) \end{bmatrix}^\top  \right] = \begin{bmatrix} V_{{x}^l{x}^l} (k)  &  V_{x^lu^l} (k)\\ V_{u^lx^l} (k) & V_{u^lu^l} (k) \end{bmatrix}.
\end{align*}
The sparsity of $\bar{V}$ is due to block-independency of the vector $\bar{z}$ and due to presence of zero-mean Gaussian noise.\\
Finally, recalling \eqref{eq: main NCS}, for the sake of compactness, $A$ and $B$ are partitioned  as
\begin{align*}
  A  = \begin{bmatrix} A_1 \vert A_2  \end{bmatrix}, \quad &B = \begin{bmatrix} B_1 \vert  B_2 \end{bmatrix}. 
 \end{align*}
where $A_1 \in \mathbb{R}^{n \times n_1}$, $A_2 \in \mathbb{R}^{n \times n_2}$, $B_1 \in \mathbb{R}^{n \times m_1}$ and $B_2 \in \mathbb{R}^{n \times m_2}$. Similarly, referring to \eqref{eq:matrixQ}, matrix $Q$ is partitioned as 
\begin{align*}
    Q = [Q_{x_1} \vert Q_{x_2}  \vert Q_{u_1} \vert Q_{u_2}]
    \end{align*}
where $Q_{x_1} \in \mathbb{R}^{(n + m)\times n_1}$, $Q_{x_2} \in \mathbb{R}^{(n + m)\times n_2}$, $Q_{u_1} \in \mathbb{R}^{(n + m)\times m_1}$, and $Q_{u_2} \in \mathbb{R}^{(n + m)\times m_2}$. Furthermore, we define the following two matrices
\begin{align*} 
    Q^1 = [Q_{x_1} \vert Q_{u_1} ],\quad Q^2 = [Q_{x_2} \vert Q_{u_2} ].
\end{align*}
}
\subsection{Equivalent Problem Formulation}
In order to rewrite the constraints appearing in equation \eqref{eq: problem2} as a function of $\bar{V}$, vectors $x(k)$, $u(k)$, $z(k)$ are obtained pre-multiplying the new variable $\bar {z}(k)$ according to
\begin{align*}
    \numberthis\label{eq: extraction of state}
    x(k) = C_x \bar{z}(k), &\qquad u(k) = C_u\bar{z}(k), & \quad z(k) = C \bar{z}(k).
\end{align*}
where
\begin{align*}
    C = \begin{bmatrix} C_x\\ C_u \end{bmatrix} = \left[
    \begin{array}{c|c|c}
         \begin{aligned} I \quad 0 \end{aligned} & \begin{aligned} I \quad 0 \\ 0 \quad 0 \end{aligned} &  \begin{aligned} 0 \quad 0 \\ I \quad 0    \end{aligned}\\
         \hline
         \begin{aligned} 0 \quad I \end{aligned} & \begin{aligned} 0 \quad I \\ 0 \quad 0 \end{aligned} &  \begin{aligned} 0 \quad 0 \\ 0 \quad I    \end{aligned}
    \end{array}\right].    
\end{align*}
The evolution of the original state $x(k)$ expressed as a function of $\bar{z}(k)$ is now
\begin{equation}
    x(k+1) = \begin{bmatrix} A & B  \end{bmatrix} C \,\bar{z}(k) + w(k).
    \label{eq: fake evolution}
\end{equation}
Combining the expressions in equations \eqref{eq: augmented state variance}, \eqref{eq: extraction of state} and \eqref{eq: fake evolution} the variance of the state $x$ can be written as
\begin{align*}
	\numberthis \label{eq: variance xx}
	V_{xx} (k)&= {\rm E}\left[ x(k) x(k)^\top \right] \\
	& = {\rm E}\left[ (C_x \bar{z}(k))(C_x \bar{z}(k))^\top \right] = C_x \bar{V}(k)C_x^\top
\end{align*}
In the same way the variance of input $u(k)$ equals
\begin{align*}
	\numberthis \label{eq: variance uu}
	V_{uu} (k)&= {\rm E}\left[ u(k) u(k)^\top \right]=  C_u \bar{V}(k)C_u^\top
\end{align*}
Finally, from \eqref{eq: global system} and \eqref{eq: variance xx}, the evolution of the system's state imposes the following recursive covariance equation
\begin{align*}
	\numberthis\label{eq: variance deco}
	& C_x \bar{V}(k+1)C_x^\top = V_{xx} (k+1)={\rm E}\left[ x(k+1) x(k+1)^\top \right] \\
	& = \begin{bmatrix} A & B  \end{bmatrix} C \, {\rm E} \left[\bar{z}(k)\bar{z}(k)^\top\right] C^\top\begin{bmatrix} A & B  \end{bmatrix}^\top + {\rm E} \left[{w}(k)w(k)^\top\right] \\
	& = \begin{bmatrix} A & B  \end{bmatrix} C \, \bar{V}(k)C^\top\begin{bmatrix} A & B  \end{bmatrix}^\top + \Sigma_w. \end{align*}
Similarly from the assumption on the state initial condition, the equivalent condition for the covariance is written as
\begin{align*}
	\numberthis \label{eq: variance xx initial condition}
	V_{xx} (0)&= {\rm E}\left[ x(0) x(0)^\top \right] = C_x \bar{V}(0)C_x^\top = \Sigma_x.
\end{align*}
We then have the following proposition which is the main achievement of this subsection.
\begin{proposition}
\label{eq:proposition}
Let $\bar{V}$ be the covariance of the extended vector $\bar{z}$. Problem \eqref{eq: problem1} is equivalent to
\begin{align*}
    \numberthis \label{eq: quadratic cost variafnce}
    \min_{\footnotesize \bar{V}(0:T) \succeq 0}\quad &  tr(C_{x}^\top Q_{T} C_{x} \bar{V}(T)) + \sum_{k=0}^{T-1} tr({C^\top Q C \bar{V}(k)})  \\
    \text{s.t.} \quad &   C_{x} \bar{V}(0) C_{x}^\top =\Sigma_x  \\
    & C_{x} \bar{V}(k+1) C_{x}^\top = \begin{bmatrix} A & B  \end{bmatrix}  C \bar{V}(k) C^\top \begin{bmatrix} A & B  \end{bmatrix}^\top + \Sigma_w \\
    & tr(C^\top W_i C \bar{V}(k)) \le p^{i}_k  
\end{align*}
\end{proposition}
\begin{proof}
The proof follows from problem in \eqref{eq: problem1} and equations \eqref{eq: variance deco} and \eqref{eq: variance xx initial condition}.
\end{proof}
\begin{remark}
\label{rem}
 Although the methodology is presented for the case of 2-player system, it can be extended to a system of $N$ players using  an algorithmic approach for state decomposition \cite{lamperski2015}. 
\end{remark}
\section{Information-oriented Optimization via Dual Decomposition}
\label{sec:dual}
In this section we proceed to define the dual problem to \eqref{eq: quadratic cost variafnce}, which allows to transform the original constrained problem \eqref{eq: problem1} into an unconstrained one. To this end, we introduce dual variables $S(k) \in \mathbb{R}^{n \times n}, k=0,\ldots,T,$ to account for constraints on the evolution of $\bar{V} (k)$, as defined in \eqref{eq: variance deco} and \eqref{eq: variance xx initial condition} . Additionally, dual scalar variables $\tau_i (k) \in \mathbb{R}^+\, , {k=0,\ldots,T-1},$ are defined to account for power constraints in the overall system. 
\subsection{Computation of Dual Variables}
Introducing the Langrange multipliers $S(0), \ldots, S(T)$ and $\tau_i (0), \ldots, \tau_i (T-1)$ the primal problem \eqref{eq: quadratic cost variafnce} is equivalent to 
\begin{align*}
    \max_{\footnotesize S(0:T), \tau_i (0:T-1)} & \min_{\footnotesize \bar{V}(0:T)}\quad   tr \left (S(0)(\Sigma_x - C_x \bar{V}(0) C_{x}^\top) \right)  \\
    & + tr(C_{x}^\top Q_{T} C_{x} \bar{V}(T)) + \sum_{k=0}^{T-1} tr({Q C \bar{V}(k) C^\top})\\
    & +\sum_{k=0}^{T-1} tr \left( S(k+1) \begin{bmatrix} A & B  \end{bmatrix}  C \bar{V}(k) C^\top \begin{bmatrix} A & B  \end{bmatrix}^\top \right)  \\
     & +\sum_{k=0}^{T-1} tr \left( S(k+1) (\Sigma_w - C_x \bar{V}(k+1) C_x^\top )  \right) \\
    & +\sum_{k=0}^{T-1} \sum_{i=1}^{M} tr \left( \tau_i(k) (C^\top W_i C \bar{V}(k) - p_k^i \right)
    \numberthis\label{eq: quadratic cost variance}
\end{align*}
where the constraints in \eqref{eq: quadratic cost variafnce} now appear as part of the cost in form of linear operators on covariance matrix $\bar{V}(k)$. Defining the Hamiltonian of the system
\begin{align*}
    &H(T) = tr \{ C_x^T \left( Q_{T} - S(T) \right)C_x \bar{V}(T)\}\\%To check. Not sure
    &H(k) = tr \{C^\top (Q + \begin{bmatrix} A & B  \end{bmatrix}^\top  S(k+1)  \begin{bmatrix} A & B  \end{bmatrix} + \\
    & - \begin{bmatrix} S(k) & 0 \\ 0 & 0 \end{bmatrix} + \sum_{i=1}^{M} \tau_i(k) W_i)C \bar{V}(k)\}\quad \text{for}\ k=0,\ldots,T-1
\end{align*}
the dual problem in \eqref{eq: quadratic cost variance} is
rewritten as 
\begin{align*}
    \max_{\footnotesize S(0:T), \tau_i (0:T-1)} \min_{\footnotesize \bar{V}(0:T)}\quad & H(T) + \sum_{k=0}^{T-1} \{ H(k) + \Sigma_w  \, tr S(k+1)\}+ \\
    & + \Sigma_x \, tr S(0) -\sum_{k=0}^{T-1} \sum_{i=1}^{M} \tau_i(k) p_k^i.
    \label{eq: quadiance}
\end{align*}
With the boundary condition on the Hamiltonian it follows $H(T)=0$, hence $S(T)=Q_{T}$. 
The dual function is finite if and only if  
\begin{equation*}
\numberthis\label{eq:feasabil}
    Q + \begin{bmatrix} A & B  \end{bmatrix}^\top  S(k+1)  \begin{bmatrix} A & B  \end{bmatrix} -\begin{bmatrix} S(k) & 0 \\ 0 & 0 \end{bmatrix} + \sum_{i=1}^{M} \tau_i(k) W_i \succeq 0.
\end{equation*}
Since the primal problem \eqref{eq: quadratic cost variafnce} is convex and constraints are affine, Slater's condition can be relaxed. Indeed, the constraints in \eqref{eq: quadratic cost variafnce}  are composed of linear equalities and inequalities and domain of the defined cost function is open, the Slater's condition reduces to feasibility. To this end, it is easy to verify that the set of constraints in \eqref{eq: quadratic cost variafnce} defines a non-empty region. Hence, the dual problem is equivalent to the primal and is stated as
\begin{align*}
    \numberthis\label{eq:eriafnce}
    \max_{\footnotesize S(0:T), \tau_i (0:T-1)}  \quad & tr(S(0)) \Sigma_x  + \Sigma_w \sum_{k=1}^{T} trS(k) - \sum_{k=0}^{T-1}\sum_{i=1}^{M} \tau_i(k) p^i_k   \\
    \text{s.t.} \quad &   Q(k)\\
    & + \begin{bmatrix} A^\top S(k+1) A - S(k) & A^\top S(k+1) B \\ B^\top S(k+1) A & B^\top S(k+1) B \end{bmatrix} \succeq 0\\
    & S(T+1) = 0
\end{align*}
where the constraint in \eqref{eq:eriafnce} is obtained from \eqref{eq:feasabil} by defining
\begin{align}
Q (k) = \left \lbrace
\begin{aligned}
    &Q + \sum_{i=1}^{M} \tau_i (k) W_i,  \quad & k=1,\ldots,T-1\\
    &\begin{bmatrix} Q_T & 0 \\ 0 & 0 \end{bmatrix},   & k=T.
\end{aligned}
\right.
\end{align}
With fixed values of $\tau_i$, the previous equation is maximized for every time-instant $k$ with
\begin{align*}
    \numberthis\label{eq:recursive}
    S(k) &= A^\top S(k+1) A + Q_{xx}(k) - L(k)^\top Y(k) L(k) \\
    Y(k) &= (B^\top S(k+1) B + Q_{uu}(k)) \\
    L(k) &= Y(k)^{-1} (B^\top S(k+1)A + Q_{xu}^\top (k))
\end{align*}
which can be proved by analogously to \cite{Gattami2010}. Indeed, the choice of $S(k)$ should be made such that $tr S(k)$ is maximized and at the same time constraint in \eqref{eq:feasabil} is satisfied, under the condition that the optimal value of $S(k+1)$ is known. To this end, since any choice of $S(k)$ with trace greater than the trace of \eqref{eq:recursive} violates the constraint in \eqref{eq:feasabil}, the choice in \eqref{eq:recursive} is optimal.
The variables $\tau_i$ have to be computed numerically from \eqref{eq:eriafnce} accounting for \eqref{eq:recursive}.
\subsection{Optimal Information-constrained Control}
In this subsection we show how to obtain the solution via information decomposition. In paragraph \ref{par: covariance decompostion} we introduced state, input and covariance decomposition. In the 2-player's case, we obtain three information sets: $\mathcal{I}_0,\mathcal{I}_1,\mathcal{I}_2$, that are defined by  \eqref{eq:information set}, \eqref{eq:common information set} and referred herein as the coordinator, first subsystem and second subsystem respectively. Moreover, the coordinator is assumed to have the following information about the overall system
\begin{align*}
    \left(A_0, B_0, Q^0, x^0 (k) \right) \triangleq \left(A, B, Q, \hat{x}(k) \right).
    \end{align*}
Before stating the main result of this paper, we define the expression for $J^l, \ l=0,1,2$ as
\begin{align*}
\numberthis\label{eq: expression of J}
    {J}^l ({V}^l,S,\tau) &= tr \left( Q_{T} {F}_l^\top {V}^l(T) {F}_l\right) + \sum_{k=0}^{T-1} tr\left(Q^l  {V}^l (k)\right)  \\ 
    &+ tr\left\lbrace S(k+1) \left(\begin{bmatrix} A_l \vert B_l  \end{bmatrix}  {V}^l(k) \begin{bmatrix} A_l \vert B_l  \end{bmatrix} ^\top \right) \right\rbrace\\
     &- tr\left\lbrace S(k+1) \left( {F}_l^\top {V}^l (k+1) {F} + \frac{\Sigma_w }{3}\right) \right\rbrace\\
    &+tr \left\lbrace S(0) \left({F}_l^\top {V}^l(0) {F}_l - \frac{\Sigma_x }{3}\right)  \right\rbrace \\
    &+ \sum_{k=0}^{T-1} \sum_{i=1}^{M} tr \left( \tau_i(k) W_i  {V}^l (k) - q_k^i \right)
    \end{align*}
where ${F}_0$, $F_1$ and $F_2$ are such that
\begin{align*}
  \numberthis\label{eq: expression of F}
     {F}_0 {V}^0 (k) {F_0}^\top &=  V_{\hat{x}\hat{x}}(k),\\
     {F}_1 {V}^1 (k) {F_1}^\top &=  \begin{bmatrix} V_{\omega_1\omega_1}(k) & 0\\ 0 & 0  \end{bmatrix},\\
     {F}_2 {V}^2 (k) {F_2}^\top &=\begin{bmatrix} 0 & 0\\ 0 & V_{\omega_2 \omega_2}(k)  \end{bmatrix}.
 \end{align*}
Moreover, the definition of $q_k ^i$ is given by identity: ${p_k ^i = 3 q_k ^i}$. 
We can now state the main result of this paper.
\begin{theorem}[Information-constrained optimal control]
\label{theorem}
Let the system dynamics be given by equation \eqref{eq: global system}. Considering the optimization problem defined in \eqref{eq: problem1} and denoting by $S(k)$ and $\tau_{i} (k)$ the optimal values of the dual variables introduced in \eqref{eq: quadratic cost variance} %and computed according to \eqref{eq:eriafnce} and \eqref{eq:recursive}
we state the following.
\begin{enumerate}[$i$.]
    \item The problem \eqref{eq: problem1} is
    decoupled into the sum of independent sub-problems that are linear in the respective decision variables, i.e., it is equivalent to 
        \begin{align*}
            \numberthis\label{eq: variance hhb}
            & \sum_{l = 0}^{2}\quad \min_{\footnotesize {V}^l (0:T)} {J}^l ({V}^l (0:T),S (0:T),\tau_{1:M} (0:T-1))
        \end{align*}
        where ${J}^l$ is defined in \eqref{eq: expression of J} and $V^l$, $l=0,1,2$ are defined in \eqref{eq: augmented state variance}.
    \item The optimal covariances $V^l$, $l=0,1,2$ of \eqref{eq: variance hhb} are computed according to
        \begin{align*}
            \numberthis\label{eq: optimal Vs}
            &V^l(k) = \begin{bmatrix}
                V^l_{xx}(k) & V^l_{xu}(k) \\
                V^l_{ux}(k) & V^l_{uu}(k)
            \end{bmatrix},\\
            &{V}^l_{xx} (0) =  \frac{\Sigma_x}{3}, \\
            &{V}^l_{ux} (k) = - {L}_l (k) {V}^l_{xx} (k), \\    
            &{V}^l_{uu} (k) =  {V}^l_{ux} (k) \left({V}^l_{xx}(k)\right)^{-1}  {V}^l_{xu} (k), \\
            &{V}^l_{xx} (k+1) = \begin{bmatrix} {A}_l & {B}_l \end{bmatrix} {V}^l (k) \begin{bmatrix} {A}_l & {B}_l \end{bmatrix}^\top + \Sigma_w .
        \end{align*} 
        where ${L}_l (k)$ is 
        \begin{align*}
        \numberthis \label{eq:Ll}
            {L}_l (k) = \left({B}_l^\top S(k+1) {B}_l + {Q}^l_{uu}\right)^{-1}\left ({A}_l^\top S(k+1) {B}_l  + {Q}^l_{xu}\right)^\top.
        \end{align*} 
    \end{enumerate}
\end{theorem}
\begin{proof}[of $i.$]
From Proposition~\eqref{eq: quadratic cost variafnce}, problem~\eqref{eq: problem1}  and \eqref{eq: quadratic cost variafnce} are equivalent. Furthermore, from  equations \eqref{eq:noise1}, \eqref{eq: augmented state variance}  and \eqref{eq: quadratic cost variance} accounting for the specific structure of matrix $C_x$ one gets
\begin{align*}
 C_x \bar{V} (k) C_x ^\top &= V_{xx} (k) = V_{\hat{x}\hat{x}}(k) + \begin{bmatrix} V_{\omega_1\omega_1}(k) & 0\\ 0 & V_{\omega_2 \omega_2}(k)  \end{bmatrix}\\
   &= {F}_0 {V}^0 (k) {F_0}^\top + F_1 V^1 (k) F_1^\top + F_2 V^2 (k) F_2^\top 
\end{align*}
where ${F_0}$, $F_1$ and $F_2$ are extraction matrices since $V_{\hat{x}\hat{x}}(k)$,  $V_{\omega_1\omega_1}(k)$ and $V_{\omega_2 \omega_2}(k)$ are square submatrices of ${V}^0 (k), V^1 (k)$ and $V^2 (k)$ respectively. On the other hand, from the block-diagonal structure of $\bar{V}(k)$ and sparsity of  $C$, one obtains
\begin{align*}
     tr({Q C \bar{V}(k) C^\top})  &= tr \left( Q {V}^0(k)\right) + tr \left( Q^1 {V}^1(k)\right) + tr \left( Q^2 {V}^2(k)\right) 
\end{align*}
Analogously, we obtain
\begin{align*}
        \begin{bmatrix} A & B  \end{bmatrix}  C \bar{V} (k) C^\top \begin{bmatrix} A & B  \end{bmatrix}^\top &= 
    \begin{bmatrix} A & B  \end{bmatrix} {V}^0 (k) \begin{bmatrix} A & B  \end{bmatrix} ^\top\\ 
    & +  \begin{bmatrix} A_1 & B_1  \end{bmatrix} {V}^1 (k) \begin{bmatrix} A_1 & B_1  \end{bmatrix}^\top \\
    & + \begin{bmatrix} A_2 & B_2  \end{bmatrix} {V}^2 (k) \begin{bmatrix} A_2 & B_2  \end{bmatrix}^\top
\end{align*}
With algebraic reordering the proof of the first part is completed.
\end{proof}
\begin{proof}[of $ii.$]
The second and fifth equation of \eqref{eq: optimal Vs} stated follow respectively from the condition on the variance of the initial state and equation \eqref{eq: variance deco}. To prove the second and third equation, observe that the decoupled problems in \eqref{eq: variance hhb}  have a similar structure. Therefore, with the optimal values of $S(k)$ and $\tau_i (k)$, each problem in equation \eqref{eq: variance hhb} is written as
\begin{align*}
    \min_{\footnotesize {V}^l (0:T-1) \succeq 0}  \sum_{k=0}^{T-1} tr({Z}^l (k)  {V}^l (k)) +  \sum_{k=0}^{T} tr(S(k))
    - \sum_{k=0}^{T-1} \sum_{i=1}^{M} \tau_i (k) q^{i}_k 
\end{align*}
where ${Z}^l (k) $ is given by
\begin{align*}
  {Z}^l (k) =  \begin{bmatrix} {X}_l {Y_l}^{-1} {X_l}^\top & {X_l} \\ {X_l}^\top & {Y_l} \end{bmatrix}
\end{align*}
and the values of matrices $ {X}_l $ and $ {Y}_l $ are computed recursively
\begin{align*}
   & {X_l} = {A_l}^\top S(k+1) {B_l}  + Q^l_{xu} \\
   & {Y_l} = {B_l}^\top S(k+1) {B_l} + Q^l_{uu}.
\end{align*}
To conclude the proof, exploiting the linearity of the subproblems, in order to compute the optimal covariances $V_l$ it is sufficient to verify if the condition $ tr(Z^l (k) V^l (k)) = 0$ is satisfied for a certain choice of the covariance matrix $V_l$. Indeed
\begin{align*}
tr(Z^l (k) V^l (k))&=tr \begin{bmatrix} {X}_l {Y_l}^{-1} {X_l}^\top V_{xx} ^l + X_l V_{ux} ^l &  * \\ * & X_l ^ \top V_{xu} ^l + Y_l V_{uu} ^l  \end{bmatrix} 
\end{align*}
By imposing to the diagonal elements in latter equation to be zero and recalling the assumption on positive-definitness (and thus invertibility) of $Q_{uu} ^l$ it follows:
\begin{align*}
 V_{ux} ^l &= - Y_l ^{-1} X_l ^\top V_{xx} ^l = - L_l V_{xx} ^l  \\
 V_{uu} ^l &= - Y_l ^{-1} X_l ^T V_{xu} ^l =  V_{ux} ^l (V_{xx} ^l)^{-1} V_{xu} ^l
\end{align*} 
which concludes the proof.
\end{proof}
\begin{corollary}
 Consider the system \eqref{eq: main NCS} and the optimization problem defined in \eqref{eq: problem1}. For the 2-player system, the optimal  control law is given by
    \begin{align*}
        \numberthis\label{eq: optimal control law}
        &u(k) = u^0 (k)+\begin{bmatrix} u^1 (k) \\  u^2 (k)  \end{bmatrix} 
    \end{align*}
where  ${u}^l (k), l=0,1,2$ is computed as
        \begin{align*}
    u^l (k) = - L_l (k) x^{l} (k)
        \end{align*} 
and $L_l$ is defined by \eqref{eq:Ll}.
\end{corollary}
\begin{proof}
According to Proposition \ref{eq:proposition}, the problem defined in \eqref{eq: problem1} is equivalent to the covariance selection problem in \eqref{eq: quadratic cost variafnce}. Since the latter is decomposed in Theorem \ref{theorem} and optimal values of covariances are provided by \eqref{eq: variance hhb}, the optimal control law follows in straightforward manner. Indeed, the control inputs referring to the coordinator and two subsystems are given by
\begin{align*}
   \numberthis\label{eq:a123}
    & {u}_l (k) = -{V}^l_{ux} (k) {{V}^l_{xx}}^{-1} (k) {x}^{l} (k) = - L^l (k) x^l (k)
    \end{align*}
    where
    \begin{align*}
    &{L}_{0} (k) =  (B^\top S(k+1) B + Q_{uu}) ^{-1} (A^\top S(k+1) B  + Q_{xu})^\top\\
    &L_{1} (k) =(B_1^\top S(k+1) B_1 + Q^1 _{uu}) ^{-1} (A_1^\top S(k+1) B_1  + Q^1 _{xu})^\top\\ 
    &L_{2} (k) = (B_2^\top S(k+1) B_2 + Q^2 _{uu}) ^{-1} (A_2^\top S(k+1) B_2  + Q^2 _{xu})^\top.
\end{align*} 
\end{proof}

\subsection{Interpretation of  Control Input Structure}
Consider a 2-player network with one-step communication delay as depicted in Fig.~\ref{fig:2-player}. It can be transformed into an equivalent network by introducing a dummy node, herein referred to as coordinator ${\mathcal{C}}$ (this is illustrated in Fig.~\ref{fig:2-player with coordinator at time t}). The colocated control units of subsystems $\mathcal{S}_1$ and $\mathcal{S}_2$ are of limited computational power (e.g. they might be routers, switches etc.) and limited memory. The coordinating unit ${\mathcal{C}}$ is assumed to be able to perform more complex computations. However, as it can be noted, it also has access to limited information about the overall system - more precisely, it knows a one-step delayed information about both subsystems. 

In our approach ${\mathcal{C}}$ computes and stores the sequences of $S(0:T)$ and ${\tau_i (0:T-1)}$ offline, based on equations \eqref{eq:eriafnce} and \eqref{eq:recursive}. During the system execution, at time-instant $k$, the coordinator ${\mathcal{C}}$ sends the matrix $S(k+1)$ to the local units. Then, using equations \eqref{eq:a123}, the calculation of the gains $L_1(k)$ and $L_2(k)$ is computed locally at the control units of $\mathcal{S}_1$ and $\mathcal{S}_2$ by matrix multiplications, thus avoiding additional memory requirements. Moreover, the coordinator ${\mathcal{C}}$, computes the estimate of the overall state based on delayed knowledge, and passes the command to units $\mathcal{S}_1$ and $\mathcal{S}_2$. Hence, the corresponding inputs to be applied to the plants are computed using local measurements and the control signal from the coordinator. 
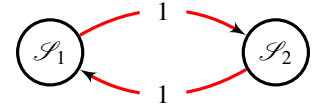
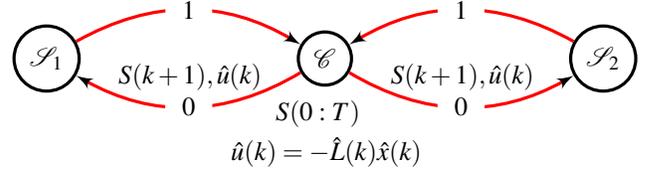
\begin{figure}[ht]
    \centering
    \begin{tikzpicture}
        \tikzset{vertex/.style = {shape=circle,draw,minimum size=1.5em, very thick}}
        \tikzset{edge/.style = {->,> = latex', draw=red,very thick}}
        % vertices
        \node[vertex] (a) at  (0,0) {$\mathcal{S}_1$};
        \node[vertex] (a1) at  (3,0) {$\mathcal{S}_2$};
        \begin{scope}[>={Stealth[black]},
            every node/.style={fill=white,circle},
            every edge/.style={draw=red,very thick}]
            \draw[edge] (a)  to[bend left] node {$1$} (a1) ;
            \draw[edge] (a1) to[bend left] node {$1$}  (a);
            \end{scope}
    \end{tikzpicture}
    \caption{2-player problem}    \label{fig:2-player}
\end{figure}
\begin{figure} [ht!]
    \centering
    \begin{tikzpicture}
        \tikzset{vertex/.style = {shape=circle,draw,minimum size=1.5em, very thick}}
        \tikzset{edge/.style = {->,> = latex', draw=red,very thick}}
        % vertices
        \node[vertex] (a) at  (0,0) {$\mathcal{S}_1$};
        \node[vertex] (a1) at  (3.7,0) {${\mathcal{C}}$};
        \node[vertex] (a2) at  (7.4,0) {$\mathcal{S}_2$};
        
       % \node[below of = a, node distance = 2cm] {};
        
        \node[below of = a1, node distance = 1cm] {$\begin{aligned}
            &\quad\ \ S(0:T)\\
            &\hat{u} (k) = - \hat{L} (k) \hat{x}(k)\\
        \end{aligned}$};
        
        \begin{scope}[>={Stealth[black]},
            every node/.style={fill=white,circle},
            every edge/.style={draw=red,very thick}]
            \draw[edge] (a)  to[bend left] node {$1$} (a1) ;
            \draw[edge] (a1) to[bend left] node (CC2) {$0$}  (a);
            \draw[edge] (a1)  to[bend right] node (CC1) {$0$}  (a2);
            \draw[edge] (a2) to[bend right] node {$1$}  (a1);
        \end{scope}
        
        \node[above of = CC2, node distance = 0.4cm] {$\begin{aligned}   S(k+1), \hat{u}(k) \end{aligned}$};
        \node[above of = CC1, node distance = 0.4cm] {$\begin{aligned}   S(k+1), \hat{u}(k) \end{aligned}$};
        
    \end{tikzpicture}
    \begin{align*}
        &\phi_1 (k) = - {L}_1 (k) {x}^1 (k) \quad & \quad \phi_2 (k) = - {L}_2 (k) {x}^2 (k)\\
        &u_1(k) = \phi_1 (k) + \begin{bmatrix}I \vert 0 \end{bmatrix}\hat{u} (k)  \quad & \quad u_2(k) = \phi_2 (k) + \begin{bmatrix}0 \vert I \end{bmatrix}\hat{u} (k) 
    \end{align*}
    \caption{Equivalent scheme at time instant $k$}\label{fig:2-player with coordinator at time t}
\end{figure}

\section{CONCLUSIONS}
\label{sec:conclude}
In this paper a framework for power-constrained optimization based on information decomposition is introduced. The linear quadratic control problem with power constraints is decomposed accordingly through covariance decomposition and Lagrangian dual reformulation. As presented, the obtained equivalent problem is linear in the new decision variables and the control gains are computed offline. The approach adopted can be extended to a network of several players.

\addtolength{\textheight}{-12cm}   % This command serves to balance the column lengths
                                  % on the last page of the document manually. It shortens
                                  % the textheight of the last page by a suitable amount.
                                  % This command does not take effect until the next page
                                  % so it should come on the page before the last. Make
                                  % sure that you do not shorten the textheight too much.

%%%%%%%%%%%%%%%%%%%%%%%%%%%%%%%%%%%%%%%%%%%%%%%%%%%%%%%%%%%%%%%%%%%%%%%%%%%%%%%%

%%%%%%%%%%%%%%%%%%%%%%%%%%%%%%%%%%%%%%%%%%%%%%%%%%%%%%%%%%%%%%%%%%%%%%%%%%%%%%%%

%%%%%%%%%%%%%%%%%%%%%%%%%%%%%%%%%%%%%%%%%%%%%%%%%%%%%%%%%%%%%%%%%%%%%%%%%%%%%%%%
%\section*{APPENDIX}

%Appendixes should appear before the acknowledgment.

%\section*{ACKNOWLEDGMENT}

%%%%%%%%%%%%%%%%%%%%%%%%%%%%%%%%%%%%%%%%%%%%%%%%%%%%%%%%%%%%%%%%%%%%%%%%%%%%%%%%

\bibliographystyle{IEEEtran}
%\bibliography{References/references}

\end{document}